\documentclass[conference]{IEEEtran}

\usepackage{cite}

\usepackage{subfigure}
\usepackage{graphicx}
\usepackage{epstopdf}

\usepackage{amsmath}
\usepackage{amssymb}
\usepackage{amsfonts}

\usepackage{algorithmic}
\usepackage{algorithm}

\usepackage{array}

\usepackage{url}

\usepackage{stfloats}
\usepackage{lineno}
\usepackage{amsthm}
\usepackage{multirow}
\usepackage{multicol}
\usepackage{lipsum}

\usepackage{latexsym,bm}
\theoremstyle{definition}
\newtheorem{theorem}{Theorem}

\theoremstyle{remark}

\begin{document}

\title{Characterizing the Energy-Efficiency Region of \\ Symbiotic Radio Communications}

\author{
\IEEEauthorblockN{
Sihan~Wang\IEEEauthorrefmark{1},
Jingran~Xu\IEEEauthorrefmark{1},
and Yong~Zeng\IEEEauthorrefmark{1}\IEEEauthorrefmark{2}
}

\IEEEauthorblockA{\IEEEauthorrefmark{1}National Mobile Communications Research Laboratory, Southeast University, Nanjing 210096, China}

\IEEEauthorblockA{\IEEEauthorrefmark{2}Purple Mountain Laboratories, Nanjing 211111, China}

Email: \{turquoise, jingran\_xu, yong\_zeng\}@seu.edu.cn
}

\maketitle

\begin{abstract}

Symbiotic radio (SR) communication is a promising technology to achieve spectrum- and energy-efficient wireless communication, by enabling passive backscatter devices (BDs) reuse not only the spectrum, but also the power of active primary transmitters (PTs). In this paper, we aim to characterize the energy-efficiency (EE) region of multiple-input single-output (MISO) SR systems, which is defined as all the achievable EE pairs by the active PT and passive BD. To this end, we first derive the maximum individual EE of the PT and BD, respectively, and show that there exists a non-trivial trade-off between these two EEs. To characterize such a trade-off, an optimization problem is formulated to find the Pareto boundary of the EE region by optimizing the transmit beamforming and power allocation. The formulated problem is non-convex and difficult to be directly solved. An efficient algorithm based on successive convex approximation (SCA) is proposed to find a Karush-Kuhn-Tucker (KKT) solution. Simulation results are provided to show that the proposed algorithm is able to effectively characterize the EE region of SR communication systems.

\end{abstract}



\IEEEpeerreviewmaketitle

\section{Introduction}

Symbiotic radio (SR) communication has been recently proposed as a promising technology to achieve both spectrum- and energy-efficient wireless communications \cite{8907447, 8254467}. The key idea of SR is to integrate passive backscatter device (BD) with active primary transmitter (PT), where the former modulates its information by backscattering the incident signal from the PT. As such, the BD reuses not only the spectrum, but also the power of the PT\cite{8692391}. Apart from enhancing the spectrum efficiency (SE) as in the conventional cognitive radio (CR) system \cite{5783948}, SR exploits the passive backscattering technology to greatly reduce the power consumption, which is expected to improve the energy efficiency (EE) significantly\cite{9193946}. Depending on the relationships of the symbol periods for the BD and the PT, the SR communication schemes can be classified as parasitic SR (PSR) and commensal SR (CSR)\cite{8907447}. In PSR, the symbol durations of BD and PT are equal, so that the backscattering transmission of BD causes interference to the PT. By contrast, in CSR, the BD symbol duration is much longer than that of the PT, so that it may contribute additional multipath signal component to enhance the active primary transmission\cite{8907447}. SR communication is anticipated to find a wide range of applications, such as E-health, wearable devices, manufacturing, and vehicle-to-everything\cite{2021arXiv211108948B}.

Significant research efforts have been recently devoted to the study of SR communications, e.g., in terms of theoretical analysis\cite{9358202,8638762} and performance optimization\cite{8665892,9461158,9036977,9686018}. In particular, as the backscattering link suffers from double-hop signal attenuations, it is typically much weaker than the direct primary link\cite{9193946}. There are several research efforts to address this issue\cite{9686018,9481926,9345739,9154299,9538923}. For example, a SR system with massive number of BDs is studied in \cite{9686018}, so as to drastically enhance the additional multipaths created by the BDs to enable full mutualism of SR. Besides, various techniques such as the reconfigurable intelligent surface (RIS) aided\cite{9481926, 9345739} or active-load assisted SR communications are also studied to strengthen the backscattering links\cite{9154299}. In \cite{9538923}, the authors study a cell-free SR system with distributed access points (APs), where an efficient channel estimation method based on two-phase uplink training is proposed.

On the other hand, the expansion of wireless networks has led to the explosive rise of energy consumption and carbon footprint\cite{5730522}. As a result, EE, which is usually defined as the ratio of the communication rate to the power consumed and has the unit bits/Joule, has been extensively studied for wireless communication over the past decade\cite{6065681}. However, there are very limited works on EE study for SR systems. In \cite{9461158}, the authors propose a Dinkelbach-based iterative algorithm to maximize the EE of a single-input single-output (SISO) SR system. In \cite{9036977}, the EE of a multiple-input single-output (MISO) SR system with finite-block-length channel codes is studied. However, such existing EE studies of SR systems mainly focus on the so-called global EE, defined as the ratio of the weighted sum-rate of PT and BD to the total power consumption, which can hardly characterize the EEs of individual devices that have more practical relevance since each device typically has its own power supply.

To fill the above gap, in this paper, we aim to characterize the EE region of a MISO SR system, which is defined as all the achievable EE pairs by the active PT and passive BD. To this end, we first derive the maximum individual EE of the PT and BD, respectively, and reveal that there exists a nontrivial trade-off between these two EEs. To study such a trade-off, an optimization problem is formulated to optimize the transmit beamforming of the PT to find the Pareto boundary of the EE, which is the union of all EE pairs for which it is impossible to increase one without decreasing the other\cite{8316986}. The formulated problem is non-convex, and we propose an efficient successive convex approximation (SCA) based algorithm to find a Karush-Kuhn-Tucker (KKT) solution. Simulation results are provided to show that the proposed algorithm is able to effectively characterize the EE regions of SR communication systems with different cross-correlations between the primary and backscattering links.

\section{System Model}

As shown in Fig. \ref{fig_1}, we consider a MISO SR communication system, which consists of an active PT, a passive BD and a primary receiver (PR). Both the PT and BD wish to communicate with the PR. The PT has $M \geq 1$ antennas and both BD and PR are equipped with one single antenna. The PT actively transmits its information-bearing signal to the PR via multi-antenna beamforming. Meanwhile, the BD communicates with the PR via backscattering communication, by modulating its own information over the incident signal from the PT. Thus, the BD reuses not only the spectrum, but also the power of the PT to transmit its own information.
\vspace{-2.5ex}
\begin{figure}[ht]
  \centering
  \includegraphics[width = .35\textwidth]{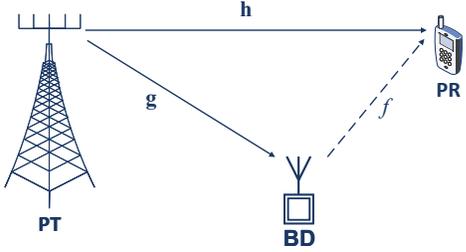}
  \vspace{-2ex}
  \caption{A MISO SR communication system for EE region characterization.}
  \label{fig_1}
  \vspace{-2ex}
\end{figure}

Denote the direct-link MISO channel from the PT to the PR as ${{\bf{h}}}={\left[{{h_{1}},...,{h_{M}}}\right]^T} \in{\mathbb{C}^{M \times 1}}$, where $h_{m}$ denotes the channel coefficient between the $m$-th antenna of the PT and the PR. Further denote the MISO channel between the PT and BD as ${{\bf{g}}}={\left[{{g_{1}},...,{g_{M}}}\right]^T}\in{\mathbb{C}^{M \times 1}}$ and that between the BD and PR as $f\in{\mathbb{C}}$. Therefore, the cascaded channel of the backscattering link is ${\bf{g}}f$.
We focus on the PSR setup\cite{8907447}, where the symbol durations of the PT and BD are equal. Let $s\left(n\right)$ and $c\left(n\right)$ denote the information-bearing symbols of the PT and BD, respectively, which follow the standard circularly-symmetric-complex-Gaussian (CSCG) distributions, i.e., $s\left(n\right) \sim \mathcal{CN}\left({0,1}\right), c\left(n\right)\sim \mathcal{CN}\left({0,1}\right)$. Further denote the transmit beamforming vector of the PT as ${\bf{w}}\in{{\mathbb C}^{M \times 1}}$, which satisfies ${\left\| {\mathbf{w}} \right\|^2} \le P_{\max}$, with $P_{\max}$ denoting the maximum allowable transmit power of the PT. Then the received signal at the PR is
\begin{small}
\begin{equation}\label{eq1}
  y\left( n \right) = {{\bf{h}}^H}{\bf{w}}s\left( n \right) + \sqrt \rho  c\left( n \right)f{{\bf{g}}^H}{\bf{w}}s\left( n \right) + z\left( n \right),
\end{equation}
\end{small}%
where $\rho\in \left[0,1\right]$ is the power reflection coefficient of the BD, $z\left(n\right)$ is the additive white Gaussian noise (AWGN) with mean zero and power ${\sigma^2}$, i.e., $z\left(n\right)\sim\mathcal{CN}\left({0,{\sigma^2}}\right)$.

Based on (\ref{eq1}), the PR first decodes the primary signal $s\left(n\right)$ by treating the BD signal $c\left(n\right)$ as noise, and then cancels the decoded primary signal before decoding the BD signal.
Therefore, the signal-to-interference-plus-noise ratio (SINR) for decoding $s\left(n\right)$ at the PR is
\begin{small}
\begin{equation}\label{eq2}
  {\gamma _s} = \frac{{{{| {{{\bf{h}}^H}{\bf{w}}} |}^2}}}{{\rho {{| f |}^2}{{| {{{\bf{g}}^H}{\bf{w}}} |}^2} + {\sigma ^2}}}.
\end{equation}
\end{small}%
Accordingly, the achievable communication rate of the PT is
\begin{small}
\begin{equation}\label{eq3}
  {R_s} = B{\log _2}\left( {1 + {\gamma _s}} \right),
\end{equation}
\end{small}%
where $B$ is the system bandwidth.
Before decoding the BD signal $c\left(n\right)$, we assume that the first term of (\ref{eq1}) is perfectly removed. Therefore, given the primary signal $s\left(n\right)$, the signal-to-noise ratio (SNR) for decoding the BD signal is
\begin{small}
\begin{equation}\label{eq4}
  {\gamma _c} {\left(s\right)}= \frac{{\rho {{| f |}^2}{{| {{{\bf{g}}^H}{\bf{w}}} |}^2}{{| {s\left( n \right)} |}^2}}}{{{\sigma ^2}}},
\end{equation}
\end{small}%
which is a random variable dependent on the primary signal $s\left(n\right)$.
By taking the expectation over $s\left(n\right)$, the average communication rate of the backscattering link is
\begin{small}
\begin{equation}\label{eq5}
  {R_c} = B{\mathbb{E}_s}\left[ {{{\log }_2}\left( {1 + {\gamma _c}{\left(s\right)}} \right)} \right].
\end{equation}
\end{small}%
With $s\left(n\right)\sim\mathcal{CN}\left( {0,1} \right)$, the squared envelope ${\left| {s\left( n \right)} \right|^2}$ follows an exponential distribution with probability density function (PDF) given by $f\left( x \right) = {e^{ - x}},x > 0$. Therefore, the average rate of the BD in (\ref{eq5}) can be expressed as\cite{8907447}
\begin{small}
\begin{equation}\label{eq6}
  {R_{c}} = B\int_0^{+\infty}{{e^{-x}}{{\log}_2}\left({1+\gamma x}\right)dx}  = - B {e^{\frac{1}{\gamma}}}{\text{Ei}} ({-\frac{1}{\gamma}}){\log_2}e,
\end{equation}
\end{small}%
where $\gamma = \frac{{\rho {{| f |}^2}{{| {{{\bf{g}}^H}{\bf{w}}} |}^2}}}{{{\sigma ^2}}}$ is the average received SNR of the backscattering link, and ${\text{Ei}}\left(x\right) \triangleq \int_{-\infty }^x {\frac{1}{t}{e^t}dt}$ is the exponential integral function.

In terms of the power consumption, the main component for the PT includes the circuit power, denoted as ${P_s}$, and the power consumed by the power amplifier, which is modelled to be proportional to the signal transmission power ${\left\| {\mathbf{w}} \right\|^2}$\cite{9036977}. Therefore, the total power consumption of the PT is ${P_s}+\mu {\left\| {\mathbf{w}} \right\|^2}$, where $\mu>1$ denotes the inefficiency of the power amplifier at the PT. Thus, the EE of the PT in bits/Joule can be expressed as
\begin{small}
\begin{equation}\label{eq7}
{EE_{PT}}
   = \frac{{B{{\log }_2}\Big( {1 + \frac{{{{| {{{\bf{h}}^H}{\bf{w}}} |}^2}}}{{\rho {{| f |}^2}{{| {{{\bf{g}}^H}{\bf{w}}} |}^2} + {\sigma ^2}}}} \Big)}}{{\mu {{{\| {\mathbf{w}} \|}^2}} + P_s}}.
\end{equation}
\end{small}%
On the other hand, since the BD does not actively transmit signal, its main power consumption is the circuit power, which is denoted as $P_c$. Thus, the EE of the BD is
\begin{small}
\begin{equation}\label{eq8}
EE_{BD} = \frac{R_c}{P_c} = - \frac{{B{{\log }_2}e}}{{P_c}}{e^{\frac{\sigma ^2}{\rho {{| f |}^2}{{| {{{\bf{g}}^H}{\bf{w}}} |}^2}}}}{\text{Ei}}\Big( { - \frac{\sigma ^2}{{\rho {{| f |}^2}{{| {{{\bf{g}}^H}{\bf{w}}} |}^2}}}} \Big).
\end{equation}
\end{small}%

\section{Maximum Individual EE of PT and BD}

It is observed from (\ref{eq7}) and (\ref{eq8}) that the EEs of the PT and BD for SR critically depend on the transmit beamforming vector $\bf{w}$. By varying $\bf{w}$, the complete EE region can be obtained. Therefore, the EE region of the considered SR system can be expressed as the union of all EE pairs $\left({EE}_{PT},{EE}_{BD} \right)$ as
\begin{small}
\begin{equation}\label{eq15}
  {\cal H }= \bigcup\limits_{{{\left\| {\bf{w}} \right\|}^2} \le {P_{\max }}} {\left\{ {\left( {EE_{PT},EE_{BD}} \right)} \right\}}.
\end{equation}
\end{small}%
Before characterizing the EE region, we first derive the maximum individual EE of the PT and BD, respectively.

For convenience, we decompose the transmit beamforming vector as ${\bf{w}}=\sqrt p {\bf{v}}$, where $p={\left\|{\bf{w}} \right\|^2}$ is the transmit power and $\bf{v}$ denotes the transmit direction with ${\left\| {\bf{v}} \right\|} = 1$. Moreover, define $\hat {\bf{h}} = \frac{{\bf{h}}}{\sigma }$ and $\hat {\bf{g}} = \frac{{\sqrt \rho  f{\bf{g}}}}{\sigma }$. Therefore, if the objective is to maximize the EE of the PT without considering that of the BD, we have the following optimization problem
\begin{small}
\begin{equation}\label{eq19}
  \begin{split}
  \mathop {\max }\limits_{{\bf{v}},p} \quad  & {EE}_{PT} = \frac{{B{{\log }_2}\Big( {1 + \frac{{p{{\bf{v}}^H}\hat{\bf{h}}{\hat{\bf{h}}^H}{\bf{v}}}}{{p {{\bf{v}}^H}{\hat{\bf{g}}}{\hat{\bf{g}}^H}{\bf{v}} + 1}}} \Big)}}{{\mu p + P_s}}\\
  {\text{s.t.}} \quad & 0 \le p \le {P_{\max }}, {\left\|{\bf{v}}\right\|}=1.
  \end{split}
\end{equation}
\end{small}%

It is not difficult to see that with any given $p$, the optimal beamforming direction ${\bf{v}}^*$ to problem (\ref{eq19}) is obtained by solving the following optimization problem
\begin{small}
\begin{equation}\label{eq9}
  \mathop {\max}\limits_{{\bf{v}}, {\left\|{\bf{v}}\right\|}=1}\quad {{\gamma _s}\left(p\right)}= \frac{{{\bf{v}}^H}{\hat{\bf{h}}}{\hat{\bf{h}}^H}{\bf{v}}} {{{{\bf{v}}^H}\Big({\hat{\bf{g}}}{\hat{\bf{g}}^H}+ {\frac{1}{p}}{{\bf{I}}_M} \Big)}{\bf{v}}}
\end{equation}
\end{small}%
Problem (\ref{eq9}) is the standard Rayleigh quotient problem, whose optimal solution is given by
\begin{small}
\begin{equation}\label{eq20}
{{\bf{v}}^*} \left(p\right) = \frac{{{{\Big({\hat{\bf{g}}}{\hat{\bf{g}}^H} + \frac{1}{p}{{\bf{I}}_M} \Big)}^{ - 1}}{\hat{\bf{h}}}}}{{\| {{{({\hat{\bf{g}}}{\hat{\bf{g}}^H} + \frac{1}{p}{{\bf{I}}_M} )}^{ - 1}}{\hat{\bf{h}}}} \|}},
\end{equation}
\end{small}%
and the corresponding maximum value of (\ref{eq9}) is
\begin{small}
\begin{equation}\label{eq39}
  {\gamma_s}^* \left(p\right) = {\hat{\bf{h}}^H}{{\Big( \hat{{\bf{g}}}{\hat{\bf{g}}^H} + \frac{1}{p}{{\bf{I}}_M} \Big)}^{ - 1}}{\hat{\bf{h}}}.
\end{equation}
\end{small}%
By substituting the SINR (\ref{eq39}) into (\ref{eq19}), the optimization problem reduces to finding the optimal transmit power $p^*$ as:
\begin{small}
\begin{equation}\label{eq41}
  \begin{split}
  \mathop {\max }\limits_p \quad & {EE}_{PT}\left(p\right) = \frac{{B{{\log }_2}\Big( {1 + {{\hat {\bf{h}}}^H}{{( {\hat {\bf{g}}{{\hat {\bf{g}}}^H} + \frac{1}{p}{{\bf{I}}_M}} )}^{ - 1}}\hat {\bf{h}}} \Big)}}{{\mu p + {P_s}}}\\
  {\text{s.t.}} \quad & 0 \le p \le {P_{\max }}.
  \end{split}
\end{equation}
\end{small}%
Note that for the special case when ${\hat {\bf{g}}}=0$, problem (\ref{eq41}) reduces to the EE maximization problem of the conventional point-to-point channel, whose optimal solution is obtained in \cite{5783982,20152400928493}. For the more general problem (\ref{eq41}), the optimal solution is obtained in the following Theorem.

\begin{theorem}
   The optimal solution $p^*$ to problem (\ref{eq41}) is ${p^*} = \min \left\{ {{p_0},{P_{\max }}} \right\}$, where $p_0$ is the unique and positive root of the following equation with respect to $p$:
   \begin{small}
   \begin{equation}\label{eq34}
     {\frac{{\left( {\mu p + P_s} \right)f'\left( p \right)}}{1+{f\left(p\right)}} - \mu \ln \left(1+ {f\left( p\right)} \right)} = 0,
   \end{equation}
   \end{small}%
   where
   \begin{small}
   \begin{equation}\label{eq35}
     f\left( p \right) = \frac{{{{\| {\hat {\bf{h}}} \|}^2}p + \left( {{{\| {\hat {\bf{h}}} \|}^2}{{\| {\hat {\bf{g}}} \|}^2} - {{| {{{\hat {\bf{g}}}^H}\hat {\bf{h}}} |}^2}} \right){p^2}}}{{1 + {{\| {\hat {\bf{g}}} \|}^2}p}}.
   \end{equation}
   \end{small}%
\end{theorem}
\begin{proof}
   Please refer to the appendix.
\end{proof}

Note that equation (\ref{eq34}) can be solved numerically or via efficient bisection method. Therefore, the maximum $EE_{PT}$, denoted by $\eta_{PT}^{\left(1\right)}$, is obtained as
\begin{small}
\begin{equation}\label{eq21}
  {\eta _{PT}^{\left(1\right)}} = \frac{{B{{\log }_2}\Big( {1 + {{\hat {\bf{h}}}^H}{{( {\hat {\bf{g}}{{\hat {\bf{g}}}^H} + \frac{1}{p^*}{{\bf{I}}_M}} )}^{ - 1}}\hat {\bf{h}}} \Big)}}{{\mu {p^*} + P_s}}.
\end{equation}
\end{small}%
Furthermore, by substituting the solution to (\ref{eq8}), the resulting EE of the BD is
\begin{small}
\begin{equation}\label{eq22}
  {\eta _{BD}^{\left(1\right)}} =  - \frac{{B{{\log }_2}e}}{{{P_c}}}\exp ( {\frac{1}{{\gamma _c^*}}} ){\text{Ei}}( { - \frac{1}{{\gamma _c^*}}} ),
\end{equation}
\end{small}%
where $\gamma _c^* = \frac{{{{| {{{\hat {\bf{g}}}^H}{{( {\hat {\bf{g}}{{\hat {\bf{g}}}^H} + \frac{1}{{{p^*}}}{{\bf{I}}_M}} )}^{ - 1}}\hat {\bf{h}}} |}^2}}}{{{{\| {{{( {\hat {\bf{g}}{{\hat {\bf{g}}}^H} + \frac{1}{{{p^*}}}{{\bf{I}}_M}} )}^{ - 1}}\hat {\bf{h}}} \|}^2}}}{p^*}$.

Next, we consider the case when the objective is to maximize the EE of the BD, without considering that of the PT. Based on (\ref{eq8}), the problem can be formulated as
\begin{small}
\begin{equation}\label{eq36}
  \begin{aligned}
  \mathop {\max}\limits_{{\bf{v}},p}\quad & {EE}_{BD} =  - \frac{{B{{\log }_2}e}}{{P_c}}{e^{\frac{1}{{{{| {{{\hat {\bf{g}}}^H}{\bf{v}}} |}^2}p}}}}{\text{Ei}}\Big( { - \frac{1}{{{{| {{{\hat {\bf{g}}}^H}{\bf{v}}} |}^2}p}}} \Big)\\
  {\text{s.t.}}\quad  &  0 \le p \le {P_{\max }}, {\left\|{\bf{v}}\right\|}=1.
  \end{aligned}
\end{equation}
\end{small}%

To find the optimal solution to (\ref{eq36}), we use the fact that the function $-{e^{\frac{1}{\gamma }}}{\text{Ei}}( { - \frac{1}{\gamma }} )$ is monotonically increasing with respect to $\gamma \ge 0$ \cite{8907447}. As a result, it is equivalent to maximize $|\hat {\bf{g}}^H {\bf{v}}|^2 p$, whose optimal solution is $p^*=P_{\max}$ and $\bf{v}^*=\frac{\hat{\bf{g}}}{\left\|\hat{\bf{g}} \right\|}$. Therefore, the maximum achievable EE of the BD is
\begin{small}
\begin{equation}\label{eq14}
\eta_{BD}^{\left(2\right)}=
- \frac{{B{{\log }_2}e}}{{P_c}}{e^{\frac{1}{{\| \hat{\bf{g}} \|^2}{P_{\max}}}}}{\text{Ei}}\Big( { - \frac{1}{{\| \hat {\bf{g}} \|^2}{P_{\max}}}} \Big),
\end{equation}
\end{small}%
and the resulting EE of the PT is
\begin{small}
\begin{equation}\label{eq40}
\eta_{PT}^{\left(2\right)} = \frac{{B{{\log }_2}\Big( {1 + \frac{{{P_{\max }}{{| {{{\hat {\bf{h}}}^H}\hat {\bf{g}}} |}^2}}}{{{P_{\max }}{{\| {\hat {\bf{g}}} \|}^4} + {{\| {\hat {\bf{g}}} \|}^2}}}} \Big)}}{{\mu {P_{\max }} + {P_s}}}.
\end{equation}
\end{small}%

Based on the above results, it is found that maximizing ${EE}_{PT}$ and ${EE}_{BD}$ lead to completely different solutions. Specifically, to maximize the EE of the BD, the PT should use its maximum power and direct the signal towards the BD via maximal ratio transmission (MRT) over the PT-BD channel ${\hat {\bf{g}}}$. By contrast, to maximize the EE of the PT, the minimum mean square error (MMSE) transmit beamforming should be used so as to balance the effect of maximizing its desired signal via channel ${\hat{\bf{h}}}$, as well minimizing the interference power via channel $\hat{\bf{g}}$. Besides, the PT should not use the maximum power in general. Such results demonstrate that there exists a non-trivial trade-off between maximizing the EE of the PT and that of the BD. To optimally study such a trade-off, we will characterize the EE region of the considered SR system as defined in (\ref{eq15}).

\section{EE Region Characterization}

Of particular interest of the EE region in (\ref{eq15}) is its outer boundary, also known as the Pareto boundary, which is defined as the union of all EE pairs $\left( {EE_{PT},EE_{BD}} \right)$ for which it is impossible to increase one without decreasing the other\cite{8316986}.
By following similar technique as characterizing the Pareto boundary of rate region in multi-user communication systems\cite{6484995}, any EE pair on the Pareto boundary of the EE region ${\cal H }$ can be obtained via solving the following optimization problem with a given EE profile ${\boldsymbol{\alpha}} = \left( {\alpha ,1 - \alpha } \right) $:
\begin{small}
\begin{equation}\label{eq16}
\begin{split}
  \mathop {\max }\limits_{{\eta},{\mathbf{w}}} \quad & \quad \eta \\
  {\text{s.t.}} \;
  {\text{C1:}}&\;  \frac{{B{{\log }_2}\Big( {1 + \frac{{{{| {{{\hat {\bf{h}}}^H}{\bf{w}}}|}^2}}}{{{{| {{{\hat {\bf{g}}}^H}{\bf{w}}} |}^2} + 1}}} \Big)}}{{\mu {{\| {\bf{w}} \|}^2} + P_s}} \ge \alpha \eta, \\
  {\text{C2:}}&\;   - \frac{{B{{\log}_2}e}}{{P_c}}{e^{\frac{1}{{{{| {{{\hat {\bf{g}}}^H}{\bf{w}}} |}^2}}}}}{\text{Ei}}\Big({-\frac{1} {{{{| {{{\hat {\bf{g}}}^H}{\bf{w}}} |}^2}}}} \Big) \ge \left( {1 - \alpha } \right)\eta, \\
  {\text{C3:}}&\; {\left\| {\mathbf{w}} \right\|^2}\le{P_{\max }},
\end{split}
\end{equation}
\end{small}%
where $0 < \alpha < 1$ is the target ratio between $EE_{PT}$ and $ \eta $. By varying $\alpha$ between $0$ and $1$, the complete Pareto boundary of the EE region ${\cal H }$ can be characterized. For a given ${\boldsymbol{\alpha}}$, we denote ${\eta}^*$ as the optimal value of problem (\ref{eq16}). Then ${{\eta}^*}{\boldsymbol{\alpha}}$ is a Pareto optimal EE pair corresponding to the intersection between a ray in the direction of ${\boldsymbol{\alpha}}$ and the Pareto boundary of the EE region.

Problem (\ref{eq16}) is non-convex, which cannot be directly solved. In this paper, we propose an efficient algorithm for problem (\ref{eq16}) by utilizing bisection method together with SCA technique. To this end, for any fixed value $\eta $, we consider the following feasibility problem:
\begin{small}
\begin{equation}\label{eq17}
\begin{split}
  {\text{Find}} \quad \quad  & {\mathbf{w}} \\
  {\text{s.t.}}\,
  {\text{C1-1:}}& \,  B {{\log }_2}\Big( {1 + \frac{{{{| {{{\hat {\bf{h}}}^H}{\bf{w}}} |}^2}}}{{{{| {{{\hat {\bf{g}}}^H}{\bf{w}}} |}^2} + 1}}} \Big) \ge \alpha \eta ( {\mu {{\| {\mathbf{w}} \|}^2} + P_s}), \\
  {\text{C2-1:}}& \,  - \frac{{B{{\log }_2}e}}{{P_c}}{e^{\frac{1}{{{{| {{{\hat {\bf{g}}}^H}{\bf{w}}} |}^2}}}}}{\text{Ei}}\Big( { - \frac{1}{{{{| {{{\hat {\bf{g}}}^H}{\bf{w}}} |}^2}}}} \Big) \ge ( {1 - \alpha } )\eta,  \\
  {\text{C3:}}\;\; &\, {\left\| {\mathbf{w}} \right\|^2} \le {P_{\max }}. \\
\end{split}
\end{equation}
\end{small}%
If $\eta $ is feasible to problem (\ref{eq17}), then the optimal value of problem (\ref{eq16}) satisfies ${\eta}^* \ge \eta $; otherwise, ${\eta}^* < \eta $. Thus, by solving problem (\ref{eq17}) with different $\eta $ and applying the efficient bisection method, problem (\ref{eq16}) can be solved.

However, the feasibility problem (\ref{eq17}) is still non-convex since the constraints C1-1 and C2-1 are non-convex. To tackle such an issue, it is noted that problem (\ref{eq17}) is feasible if and only if $\eta_{BD}^* \ge \left( {1 - \alpha } \right)\eta$, where $\eta_{BD}^*$ is the optimal value of the following optimization problem
\begin{small}
\begin{equation}\label{eq37}
\begin{split}
  \mathop {\max }\limits_{\bf{w}} \; & \quad - \frac{{B{{\log }_2}e}}{{P_c}}{e^{\frac{1}{{{{| {{{\hat {\bf{g}}}^H}{\bf{w}}} |}^2}}}}}{\text{Ei}}\Big( { - \frac{1}{{{{| {{{\hat {\bf{g}}}^H}{\bf{w}}} |}^2}}}} \Big)\\
  {\text{s.t.}}\quad &
  {\text{C1-1}}, {\text{and C3}}.
\end{split}
\end{equation}
\end{small}%

Since $-{e^{\frac{1}{\gamma }}}{\text{Ei}}( { - \frac{1}{\gamma }} )$ is a monotonically increasing function with respect to $\gamma \ge 0$ \cite{8907447}, problem (\ref{eq37}) can be equivalently transformed into
\begin{small}
\begin{equation}\label{eq38}
\begin{split}
  \mathop {\max }\limits_{\bf{w}}  \; & \quad {| {{{\hat {\bf{g}}}^H}{\bf{w}}} |^2} \\
  {\text{s.t.}}\quad &
  {\text{C1-1}}, {\text{and C3}}.
\end{split}
\end{equation}
\end{small}%
As a result, problem (\ref{eq17}) is feasible if and only if $ \gamma^* \ge \gamma \left( \eta  \right) $, where $\gamma^*$ is the optimal value of problem (\ref{eq38}) and $\gamma \left( \eta \right)$ is the root of the following equation:
\begin{small}
\begin{equation}\label{eq33}
    ( {1 - \alpha } )\eta  + \frac{{B{{\log }_2}e}}{{P_c}}{e^{\frac{1}{\gamma }}}{\text{Ei}}( { - \frac{1}{\gamma }} ) = 0,
\end{equation}
\end{small}%
which can be efficiently obtained numerically for any given $\eta$.

Therefore, the remaining problem is to solve the optimization problem (\ref{eq38}), which is non-convex. Fortunately, the efficient SCA technique\cite{8918497,SCA} can be used for (\ref{eq38}).
To this end, by defining ${{\bf{H}}}={\hat{\bf{h}}}{\hat{\bf{h}}^H}$, ${{\bf{G}}}={\hat{\bf{g}}}{\hat{\bf{g}}^H}$, problem (\ref{eq38}) can be transformed into:
\begin{small}
\begin{equation}\label{eq18}
  \begin{aligned}
  \mathop {\max }\limits_{\bf{w}} \quad \; \quad  & {{\bf{w}}^H}{\bf{Gw}} \\
  {\text{s.t.}}\;
  {\text{C1-2}}:& \;\alpha \eta ( {\mu {{\left\| {\mathbf{w}} \right\|}^2} + P_s} ) +B{\log _2}( {{{\mathbf{w}}^H}{\mathbf{Gw}} + 1} )    \\
  & -B{\log _2}( {{{\mathbf{w}}^H}{\mathbf{Hw}} + {{\mathbf{w}}^H}{\mathbf{Gw}} + 1}) \le 0, \\
  {\text{C3-1}}: &  \; {\left\| {\mathbf{w}} \right\|^2} -{P_{\max }} \le  0.
\end{aligned}
\end{equation}
\end{small}%
Furthermore, by introducing a slack variable $S$, problem (\ref{eq18}) can be equivalently written as
\begin{small}
\begin{equation}\label{eq30}
 \begin{aligned}
   \mathop {\max}\limits_{{\bf{w}},S} \quad \; \quad & {{\bf{w}}^H}{\bf{Gw}} \\
  {\text{s.t.}}\;
  {\text{C1-3}}:
  & \; \alpha \eta ( {\mu {{\left\| {\bf{w}} \right\|}^2} + P_s} ) + B{\log _2}( {S + 1} )\\
  & - B{\log_2}({{{\bf{w}}^H}{\bf{Hw}}+{{\bf{w}}^H}{\bf{Gw}}+1}) \le 0,\\
  {\text{C3-1}}:
  & \; {\left\| {\bf{w}} \right\|^2} - {P_{\max }} \le 0,\\
  {\text{C4}}:\;
  & \; {{\bf{w}}^H}{\bf{Gw}} - S \le 0.
 \end{aligned}
\end{equation}
\end{small}%
It is not difficult to see that there always exists an optimal solution to (\ref{eq30}) so that the constraint ${\text{C4}}$ is satisfied with equality. To show this, suppose that $\left( {{{\bf{w}}_1},{S_1}} \right)$ is an optimal solution to (\ref{eq30}) for which C4 is satisfied with strict inequality, i.e., ${\bf{w}}_1^H{\bf{Gw}}_1 - {S_1} < 0$. Then we can further reduce $S_1$ until C4 is satisfied with equality, which is denoted as ${S_2} = {\bf{w}}_1^H{\bf{Gw}}_1$. It is obvious that the solution $\left( {{{\bf{w}}_1},{S_2}} \right)$ satisfies all constraints in (\ref{eq30}), and achieves the same objective value as the assumed optimal solution $\left( {{{\bf{w}}_1},{S_1}} \right)$. Therefore, it is also an optimal solution to (\ref{eq30}). This shows the equivalence between problem (\ref{eq30}) and (\ref{eq18}).

Problem (\ref{eq30}) can be handled by the SCA technique, which transforms the non-convex optimization problem into a series of convex optimization problems, with guaranteed convergence to a KKT solution under some mild conditions.
To that end, by using the fact that any convex differentiable function is globally lower-bounded by its first-order Taylor expansion, we have the following lower bound at a given local point ${{\mathbf{w}}^{(i)}}$
\begin{small}
\begin{equation}\label{eq28}
 \begin{aligned}
  {{{\bf{w}}^H}{\bf{Gw}}}
  & \ge {{{\bf{w}}^{(i)}}^H{\bf{G}}{{\bf{w}}^{(i)}} + 2{\mathop{\text {Re}}\nolimits} \{ {{{\bf{w}}^{(i)}}^H{\bf{G}}( {{\bf{w}} - {{\bf{w}}^{(i)}}} )} \}}  \\
  & \buildrel \Delta \over = {\phi_{lb}}( {\bf{w}}| {\bf{w}}^{(i)}).
  \end{aligned}
\end{equation}
\end{small}%
Similarly,
\begin{small}
\begin{equation}\label{eq29}
  \begin{aligned}
 {{\bf{w}}^H}{\bf{Hw}}
  & \ge {{\bf{w}}^{(i)}}^H{\bf{H}}{{\bf{w}}^{(i)}} + 2{\mathop{\text {Re}}\nolimits} \{ {{{\bf{w}}^{(i)}}^H{\bf{H}}( {{\bf{w}} - {{\bf{w}}^{(i)}}} )} \} \\
  & \buildrel \Delta \over = {\chi _{lb}}( {\bf{w}}| {\bf{w}}^{(i)}).
  \end{aligned}
\end{equation}
\end{small}%
Furthermore, we have the global upper bound for the concave function ${\log _2}\left( {S + 1} \right)$ at a given local point $S^{\left(i \right)}$
\begin{small}
\begin{equation}\label{eq31}
  \begin{aligned}
  {\log _2}( {S + 1} )
  & \le {\log _2}( {{S^{(i)}} + 1} ) + \frac{{{{\log }_2}e}}{{{S^{(i)}} + 1}}({S-{S^{(i)}}}) \\
  & \buildrel \Delta \over = {\zeta _{ub}}(S | S ^ {(i)}).
  \end{aligned}
\end{equation}
\end{small}%

Therefore, at any given local point ${{\mathbf{w}}^{(i)}}$ and $S^{\left(i \right)}$, by applying the global lower/upper bounds (\ref{eq28}) - (\ref{eq31}) to problem (\ref{eq30}), we have the following optimization problem
\begin{small}
\begin{equation}\label{eq32}
  \begin{aligned}
  \mathop {\max }\limits_{{\bf{w}},S} \quad \; \quad &  {{{\bf{w}}^{(i)}}^H{\bf{G}}{{\bf{w}}^{(i)}} + 2{\mathop{\text {Re}}\nolimits} \{ {{{\bf{w}}^{(i)}}^H{\bf{G}}( {{\bf{w}} - {{\bf{w}}^{(i)}}} )}\}}\\
  {\text{s.t.}}\,
  {\text{C1-4}}:\, & \alpha \eta ( {\mu {{\left\| {\bf{w}} \right\|}^2} + P_s}) + B{\zeta_{ub}}(S| {S}^{(i)}) -  \\
  & B {\log_2} (\phi_{lb} ({\bf{w}}| {\bf{w}}^{(i)} ) +{{\chi_ {lb}( {\bf{w}}| {\bf{w}}^{(i)})}+1}) \le 0,\\
  {\text{C3-1}}:\, &  {\left\| {\mathbf{w}} \right\|^2} -{P_{\max }} \le  0,\\
  {\text{C4}}:\;\,  &  {{\bf{w}}^H}{\bf{Gw}} - S \le 0.
  \end{aligned}
\end{equation}
\end{small}%

Problem (\ref{eq32}) is a convex optimization problem, which can be efficiently solved by standard convex optimization techniques or existing software tools such as CVX \cite{cvx}. Thanks to the global lower/upper bounds relationships in (\ref{eq28}) - (\ref{eq31}), the optimal value of problem (\ref{eq32}) gives at least a lower bound to that of problem (\ref{eq30}).
Therefore, an efficient KKT solution of the original non-convex problem (\ref{eq30}) can be obtained by iteratively solving the convex optimization problem (\ref{eq32}) with the local point $\left({{\bf{w}}^{(i)}}, {{S}^{(i)}}\right)$ updated in each iteration, which is summarized in Algorithm \ref{alg1}.
Note that the algorithm is guaranteed to converge since it can be shown that at each iteration, the objective values of (\ref{eq30}) and (\ref{eq32}) are monotonically non-decreasing, as verified in Section V. As such, the original problem (\ref{eq16}) can be solved based on the standard bisection method, together with Algorithm \ref{alg1}.
\addtolength{\topmargin}{0.01in}
\begin{algorithm}[t]
  \renewcommand{\algorithmicrequire}{\textbf{Input}:}
  \renewcommand{\algorithmicensure}{\textbf{Output}:}
  \renewcommand{\algorithmicif}{\quad \textbf{If}}
  \renewcommand{\algorithmicendif}{\quad \textbf{End if}}
  \renewcommand{\algorithmicelse}{\quad \textbf{Else}}
  \renewcommand{\algorithmicreturn}{\textbf{Return}}
  \caption{SCA Algorithm for Solving Problem (\ref{eq30})}
  \label{alg1}
  \begin{algorithmic}[1]
    \STATE \textbf{Initialization}: ${{\bf{w}}^{\left(0\right)}}$, ${S^{\left(0\right)}}$, $i=0$;
    \STATE \textbf{Repeat}
      \STATE \quad With given ${{\bf{w}}^{(i)}}$ and ${{S}^{(i)}}$, solve the convex optimization problem (\ref{eq32}), and denote the optimal solution as ${{\bf{w}}^{*(i)}}$ and ${S^{*(i)}}$;
      \STATE \quad Update ${{\bf{w}}^{\left(i+1\right)}}={{\bf{w}}^{*(i)}}$ and ${{S}^{\left({i+1}\right)}}={S^{*(i)}}$;
      \STATE \quad Update $i=i+1$;
    \STATE \textbf{Until} The fractional increase of the objective value of problem (\ref{eq32}) is below a certain threshold $\kappa $ ;
    \RETURN ${{\bf{w}}^*}={{\bf{w}}^{(i)}}$ and ${S^*}={{S}^{(i)}}$.
  \end{algorithmic}
\end{algorithm}

\section{Simulation Results}

In this section, simulation results are presented to demonstrate the effectiveness of our proposed algorithm. We assume that the PT is equipped with an uniform linear array (ULA) of $M=4$ elements, with the adjacent element separated by half wavelength.
Both the BD and PR have equal distance with the PT, which is $d_0 = 300{\text{m}}$, and the angle formed by PT-PR and PT-BD line segments is $\theta$. As such, without loss of generality, the coordinate of the PT, PR and BD locations can be represented as $\left(0,0\right)$, $\left(d_0,0\right)$ and $\left(d_0\cos{\theta},d_0\sin{\theta}\right) $, respectively. Therefore, the distance between BD and PR can be represented as ${d_1} = 2{d_0}\sin \left( {\frac{{{\theta }}}{2}} \right)$. All links are assumed to be the Rician fading channels, with the Rician K-factor $K=10\text{dB}$. The angle of departure (AoD) of the LoS component of the PT-PR channel $\bf{h}$ is $0$, and that of the PT-BD channel $\bf{g}$ is thus ${\theta}$.
Moreover, the large-scale path loss is modeled as $\beta  = {\beta _0}{d^{ - \alpha }}$, where ${\beta _0} = {\left( {\frac{\lambda }{{4\pi }}} \right)^2}$ is the reference channel gain, $\lambda$ is the wavelength, $d$ is the distance between the corresponding devices and $\alpha$ denotes the path loss exponent. We set the path-loss exponents of the PT-PR link, PT-BD link and BD-PR link as ${\alpha_{TR}}={\alpha_{TD}} =2.7$, ${\alpha_{DR}}=2.1$, respectively.
The carrier frequency is $3.5{\text{GHz}}$, the channel bandwidth $B$ is $10 \text{MHz}$ and the noise power is ${\sigma ^2} = -110 \text{dBm}$. Furthermore, the reflection coefficient is $\rho = 1$ and the inefficiency of the power amplifier of the PT is $\mu = 2.85$. The maximum transmit power is $P_{\max}=0.1 \text{W}$, and the circuit power of the BD is $P_c=0.2 \text{mW}$.

Fig. \ref{fig_7} plots the convergence of the proposed SCA Algorithm \ref{alg1} with randomly generated initial points, where the termination threshold is $\kappa = 0.1\%$. Note that two curves are plotted in Fig. \ref{fig_7}: the ``Accurate EE" corresponds to the exact EE of the BD by substituting the obtained beamforming vector $\bf{w}$ in each iteration to the EE model in (\ref{eq8}), while the ``Lower bound with Taylor approximation" corresponds to the EE of the BD calculated based on the optimal objective value of problem (\ref{eq32}). It is observed from Fig. \ref{fig_7} that with the proposed SCA-based algorithm, the EE of the BD and its lower bound increase monotonically and only a few iterations are needed for convergence.
\vspace{-2ex}
\begin{figure}[h]
  \centering
  \includegraphics[width=0.45\textwidth]{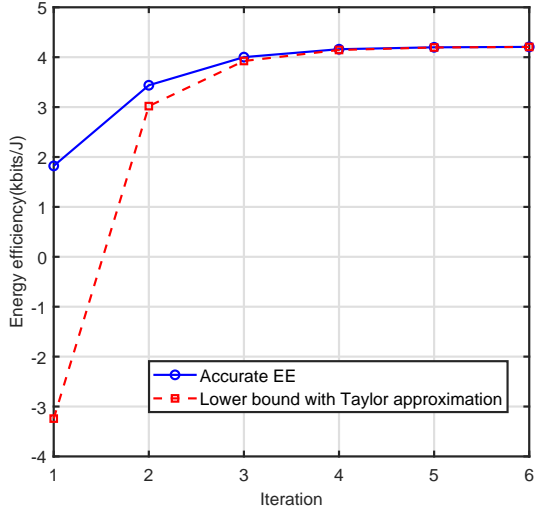}
  \vspace{-2ex}
  \caption{Convergence of the proposed SCA Algorithm \ref{alg1}.}
  \label{fig_7}
  \vspace{-2ex}
\end{figure}

\begin{figure}[h]
  \centering
  \includegraphics[width = 0.45\textwidth]{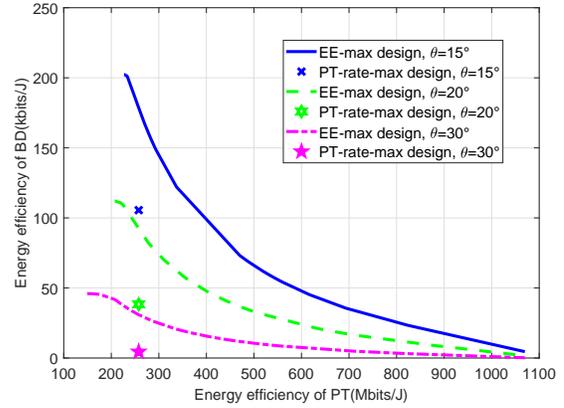}
  \vspace{-2ex}
  \caption{EE region of MISO SR with different $\theta$. The markers label the resulting EE pair with the beamforming designed to maximize the rate of the PT.}
  \vspace{-2ex}
  \label{fig_5}
\end{figure}

Fig. \ref{fig_5} plots the outer boundaries of the EE regions obtained with our proposed algorithm for different $\theta$, which are labelled as the ``EE-max design". As a benchmark comparison, Fig. \ref{fig_5} also shows the resulting EE pair for the so-called ``PT-rate-max design", for which the transmit beamforming is designed to maximize the communication rate of the PT, rather than the EE. Note that it is not difficult to see that the ``BD-rate-max design" is equivalent to maximizing the individual EE of the BD as considered in Section III. The circuit power consumption for the PT is $P_s = 20 \text{mW}$. It is firstly observed from Fig. \ref{fig_5} that there exists a non-trivial EE trade-off between PT and BD, i.e., a sacrifice of the EE for PT would lead to considerable improvement to that of the BD, and vice versa. It is also observed from Fig. \ref{fig_5} that as $\theta$ increases, the achievable EE region shrinks. This is expected since a larger $\theta$ value implies that the PT-PR channel $\bf{h}$ and PT-BD channel $\bf{g}$ are less correlated, which makes it more challenging to find the transmit beamforming to simultaneously direct significant power towards both PR and BD. Another observation from Fig. \ref{fig_5} is that the resulting EE pairs by the ``PT-rate-max design" lie in the interior of the achievable EE region. This implies that simply maximizing the communication rate is strictly energy-inefficient, which demonstrates the importance of considering EE metrics deliberately in SR systems. Similar observations are also made in Fig. \ref{fig_6}, which plots the outer boundaries of achievable EE regions and the resulting EE pairs with the ``PT-rate-max design" for two different circuit power levels $P_s=0 \text{mW}$ and $20 \text{mW}$, with $\theta=20^\circ$. It is also observed from Fig. \ref{fig_6} that as the circuit power consumption reduces, the achievable EE region enlarges, as expected.
\vspace{-2ex}
\begin{figure}[h]
  \centering
  \includegraphics[width = 0.45\textwidth]{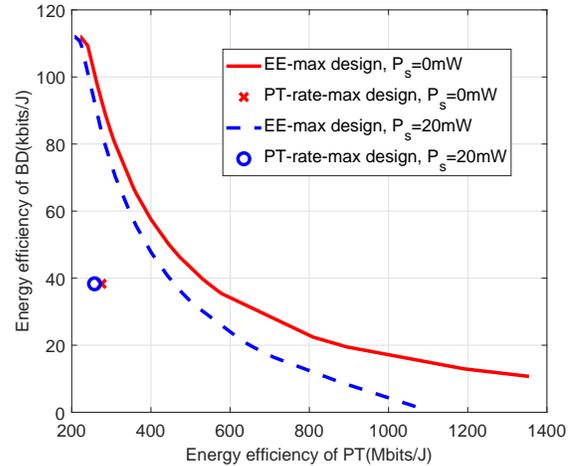}
  \vspace{-2ex}
  \caption{EE region of MISO SR with different $P_s$.}
  \vspace{-2ex}
  \label{fig_6}
\end{figure}

\section{Conclusion}

In this paper, we characterized the EE region of MISO SR systems, where the primary link shares the spectrum as well as the power with the backscattering link. We first derived the maximum individual EE of the PT and BD, respectively, and revealed that there exists a non-trivial EE trade-off between PT and BD. An optimization problem was then formulated by optimizing the transmit beamforming to find the Pareto boundary points of the EE region. An efficient SCA-based bisection algorithm was proposed to solve the formulated problem. Simulation results are provided to show the effectiveness of our proposed algorithm for EE region characterization of SR systems.

\appendix[Proof of Theorem 1]

According to the Woodbury matrix identity, the objective function ${EE}_{PT}{\left(p\right)}$ of problem (\ref{eq41}) can be written as
\begin{small}
\begin{equation}\label{eq11}
  EE_{PT}{\left(p\right)} = \frac{{B{{\log }_2}\Big( {1 + ( {{{\| {\hat {\bf{h}}} \|}^2} - \frac{{{{\left| {{{\hat {\bf{g}}}^H}\hat {\bf{h}}} \right|}^2}p}}{{1 + {{\left\| {\hat {\bf{g}}} \right\|}^2}p}}} )p} \Big)}}{{\mu p + P_s}}.
\end{equation}
\end{small}%

For convenience, denote $g\left( p \right)$ as
\begin{small}
\begin{equation}\label{eq24}
   g\left( p \right) =  EE_{PT}{\left(p\right)} = \frac{{B{{\log }_2}\left(1+ {f\left( p \right)} \right)}}{{\mu p + P_s}},
\end{equation}
\end{small}%
where
\begin{small}
\begin{equation}\label{eq23}
  f\left( p \right) =  \frac{{{{\| {\hat {\bf{h}}} \|}^2}p + \left( {{{\| {\hat {\bf{h}}} \|}^2}{{\left\| {\hat {\bf{g}}} \right\|}^2} - {{| {{{\hat {\bf{g}}}^H}\hat {\bf{h}}} |}^2}} \right){p^2}}}{{1 + {{\left\| {\hat {\bf{g}}} \right\|}^2}p}}.
\end{equation}
\end{small}%
Then the derivative of $ g\left( p \right) $ is \par
\begin{small}
\begin{equation}\label{eq25}
  g'\left( p \right) = \frac{{B{{\log }_2}e}}{{{{\left( {\mu p + P_s} \right)}^2}}}h\left( p \right),
\end{equation}
\end{small}%
where
\begin{small}
\begin{equation}\label{eq26}
  h\left( p \right) = {\frac{{\left( {\mu p + P_s} \right)f'\left( p \right)}}{1+{f\left( p \right)}} - \mu \ln \left(1+ {f\left( p \right)} \right)}.
\end{equation}
\end{small}%

It can be shown that $f\left( p \right)$ is a concave, monotonically increasing function and $f\left( p \right) >0 $, for $p \ge 0$, by checking its first- and second-order derivatives. Similarly, $h\left( p \right)$ is a monotonically decreasing function for $p\ge 0$.
Then, with $h\left( 0 \right) = {\left\| {\bf{h}} \right\|^2}P_s > 0$ and $h\left( { + \infty } \right) =  - \infty  < 0$, there exists a positive and unique zero point $p_0$ of the function $h \left( p \right)$ to satisfy
\begin{small}
\begin{equation}\label{eq27}
  h\left( p_0\right) = {\frac{{\left( {\mu p_0 + P_s} \right)f'\left( p_0 \right)}}{1+{f\left( p_0 \right)}} - \mu \ln \left(1+ {f\left( p_0 \right)} \right)} = 0.
\end{equation}
\end{small}%
Therefore, for $p \in \left[ {0,{p_0}} \right)$, $g'\left(p\right)> 0$, i.e., $g\left( p \right)$ is increasing, while for $ p \ge {p_0}$, $g'\left(p\right)\le 0$, i.e., $g\left( p \right)$ is decreasing. Since the feasible range for (\ref{eq41}) is $p\le P_{\max}$, the optimal solution is given by ${p^*} = \min \left\{ {{p_0},{P_{\max }}} \right\}$.

The proof is thus completed.

\section*{Acknowledgment}

This work was supported by the National Key R$\text{\&}$D Program of China with Grant number 2019YFB1803400.



\bibliographystyle{IEEEtran}
\bibliography{myreference}

\end{document}